%% file: main.tex
\documentclass[11pt,a4paper]{article}
\usepackage[utf8]{inputenc}
\usepackage{lmodern} 

\input{header.tex}

\title{Unifying Matrix Data Structures: \\
Simplifying and Speeding up Iterative Algorithms}

\author{
Jan van den Brand\thanks{\texttt{janvdb@kth.se}. KTH Royal Institute of Technology, Sweden.}
}

\begin{document}
	
\begin{titlepage}
	\maketitle
	
	\pagenumbering{roman}
\input{abstract.tex}

	\newpage
	\setcounter{tocdepth}{2}
	\tableofcontents
	
\end{titlepage}

\newpage
\pagenumbering{arabic}

\input{intro.tex}

\input{preliminaries.tex}

\input{matrix.tex}

\input{application.tex}

\section*{Acknowledgment}

This project has received funding from the Google PhD Fellowship Program,
and the European Research Council (ERC) under the European Unions Horizon 2020 research and innovation program 
under grant agreement No 715672.

\bibliographystyle{alpha}
\bibliography{ref}

\newpage
\appendix

\input{appendix.tex}

\end{document}

%% file: header.tex
\usepackage[bookmarks,colorlinks,breaklinks]{hyperref}
\hypersetup{urlcolor=blue, colorlinks=true, citecolor=green!50!black, linkcolor=blue}

\usepackage{multirow}
\usepackage[T1]{fontenc}
\usepackage{amsmath}
\usepackage{amsfonts}
\usepackage{amssymb}
\usepackage{amsthm}
\usepackage{thmtools}
\usepackage{pgffor}
\usepackage{cleveref}
\usepackage{geometry}
\usepackage{enumitem}
\usepackage{multirow}
\usepackage{xcolor}
\makeatletter
\renewcommand{\paragraph}{%
	\@startsection{paragraph}{4}%
	{\z@}{1.25ex \@plus 1ex \@minus .2ex}{-1em}%
	{\normalfont\normalsize\bfseries}%
}
\makeatother

\declaretheorem[numberwithin=section,refname={Theorem,Theorems},Refname={Theorem,Theorems}]{theorem}
\declaretheorem[numberlike=theorem]{lemma}

\declaretheorem[numberlike=theorem]{corollary}
\declaretheorem[numberlike=theorem]{definition}

\declaretheorem[numberlike=theorem]{task}

\newcommand{\R}{\mathbb{R}}

\newcommand{\N}{\mathbb{N}}

\newcommand{\F}{\mathbb{F}}
\renewcommand{\tilde}{\widetilde}

\newcommand{\zerovec}{\vec{0}} 
\newcommand{\zeromat}{\boldsymbol{0}} 
\newcommand{\unitvec}{\vec{e}}

\foreach \x in {A,...,Z}{%
	\expandafter\xdef\csname m\x\endcsname{\noexpand\mathbf{\x}}
}

\newcommand{\of}{\noexpand{\overline{f}}}
\newcommand{\oc}{\noexpand{\overline{c}}}
\newcommand{\ob}{\noexpand{\overline{b}}}
\newcommand{\os}{\noexpand{\overline{s}}}
\newcommand{\ot}{\noexpand{\overline{t}}}
\newcommand{\og}{\noexpand{\overline{g}}}
\newcommand{\ou}{\noexpand{\overline{u}}}
\newcommand{\ov}{\noexpand{\overline{v}}}
\newcommand{\ow}{\noexpand{\overline{w}}}
\newcommand{\ox}{\noexpand{\overline{x}}}
\newcommand{\oy}{\noexpand{\overline{y}}}
\newcommand{\oz}{\noexpand{\overline{z}}}

\foreach \x in {A,...,Z}{%
	\expandafter\xdef\csname c\x\endcsname{\noexpand\mathcal{\x}}
}

\foreach \x in {A,...,Z}{%
	\expandafter\xdef\csname om\x\endcsname{\noexpand\mathbf{\overline{\x}}}
}

\foreach \x in {A,...,Z}{%
	\expandafter\xdef\csname tm\x\endcsname{\noexpand\mathbf{\widetilde{\x}}}
}

%% file: abstract.tex
Many algorithms use data structures that maintain properties of matrices undergoing some changes.
The applications are wide-ranging and include for example matchings, shortest paths, linear programming, semi-definite programming, convex hull and volume computation.
Given the wide range of applications, the exact property these data structures must maintain varies from one application to another, forcing algorithm designers to invent them from scratch or modify existing ones.
Thus it is not surprising that these data structures and their proofs are usually tailor-made for their specific application and that maintaining more complicated properties results in more complicated proofs.

In this paper we present a unifying framework that captures a wide range of these data structures. The simplicity of this framework allows us to give short proofs for many existing data structures regardless of how complicated the to be maintained property is. 
We also show how the framework can be used to speed up existing iterative algorithms, such as the simplex algorithm.

More formally, consider any rational function $f(A_1,...,A_d)$ with input matrices $A_1,...,A_d$. We show that the task of maintaining $f(A_1,...,A_d)$ under updates to $A_1,...,A_d$ can be reduced to the much simpler problem of maintaining some matrix inverse $M^{-1}$ under updates to $M$.
The latter is a well studied problem called dynamic matrix inverse.
By applying our reduction and using known algorithms for dynamic matrix inverse we can obtain fast data structures and iterative algorithms for much more general problems.

%% file: intro.tex
\section{Introduction}
\label{sec:intro}

There exist many algorithms and data structures 
that maintain some properties of a matrix 
(or expression involving matrices) 
while the input matrices change over time.
There are many different applications and areas where such matrix data structures
are used.
Applications include graph theoretic problems like
reachability \cite{Sankowski04,BrandS19},
shortest paths \cite{Sankowski05,BrandN19},
matchings \cite{Sankowski07,MuchaS06},
but also in convex optimization, e.g.
linear programming \cite{Karmarkar84,Vaidya89a,LeeS15,cls19,lsz19,b20,JiangSWZ20,blss20}, 
cutting plane algorithms \cite{lsw15,JiangLSW20},
and $\ell_p$-norm regression \cite{adil2019iterative}.
Other areas include computational geometry
(e.g.~convex hull and volume computation \cite{FisikopoulosP16}),
statistics (e.g.~online linear least squares \cite{Gentleman74,Plackett50}), 
robotics \cite{Barfoot05}, and machine learning \cite{GaloogahiFL17}.

Each of the previously mentioned applications usually needs to maintain different properties.
For example for graph matchings one must maintain the rank of some matrix \cite{Sankowski07},
for online linear regression one needs to maintain the Penrose-pseudo inverse \cite{Gentleman74,Plackett50},
and interior point based linear program solvers require to maintain some projection matrix \cite{cls19,lsz19,b20,JiangSWZ20}.
For each application, the authors create tailor-made data structures 
that maintain the solution for their specific problem.
This results in a wide range of different data structures that all serve a different purpose.
Unfortunately, only few of the previously listed results have short and simple proofs.
As the to be maintained property grows more complicated,
the correctness and complexity analysis of these papers increases as well.
And even if some proofs are simple, the sheer amount of different data structures means that there exists many different proofs in the area.

We are able to simplify the more complicated data structures
and unify the area of matrix data structures 
by reducing a wide range of data structure problems to the so called \emph{dynamic matrix inverse} problem.
This data structure problem is about maintaining some inverse $\mM^{-1}$ while the underlying matrix $\mM$ changes over time.
Informally, our result can be stated as follows.

\begin{theorem}[Informal statement of \Cref{thm:reduction}]
\label{thm:informal_reduction}
Let $f$ be any formula with input matrices $\mA_1,...,\mA_d$ (potentially with differing dimensions),
where the formula $f$ consists of only matrix-addition, -subtraction, -multiplication, and -inversion.

Then the value of $f(\mA_1,...,\mA_d)$ can be maintained under updates to $\mA_1,...,\mA_d$ via a dynamic matrix inverse algorithm.
\end{theorem}

\Cref{thm:informal_reduction} allows for the following simplifications:
\begin{enumerate}[label=\Roman*,nosep]
\item We can simplify complicated data structures: 
For long and complex formulas $f$, one usually needs to construct rather complicated data structures. 
However, thanks to \Cref{thm:informal_reduction} it now suffices to only maintain the inverse of some matrix, 
i.e.~the simple formula $f(\mM) = \mM^{-1}$. \label{point:reduction}
\item We can unify the area of matrix data structures: 
Instead of constructing and analysing many different data structures 
it suffices to consider only dynamic matrix inverse data structures. 
This reduces the overall amount of proofs in the area. \label{point:proofs}
\item We can construct simple algorithms, if one uses existing dynamic matrix inverse algorithms as a blackbox: 
Existing data structure for dynamic matrix inverse are very efficient. 
Simple iterative algorithms can become quite fast when using these existing data structures.
So fast in fact, that the simple algorithms often match the complexity of more complicated algorithms. \label{point:iterative}
\end{enumerate}
We give some examples further below to demonstrate these claims.
These examples are more formally discussed and proven in \Cref{sec:applications},
where we also discuss some implications of \Cref{thm:informal_reduction} regarding fine-grained complexity theory.

In addition to these simplifications, \Cref{thm:informal_reduction} might also be interesting for future research. 
If one comes up with a new iterative algorithm for some problem, 
then \Cref{thm:informal_reduction} is a powerful tool to obtain a data structure for this new iterative algorithm 
by simply reducing to existing data structures.

\paragraph{Interior Point Based Linear Program Solvers}
Consider a linear program of the form
$$
\min_{\mA x=b, x\ge0} c^\top x
$$
for some $\mA \in \R^{d\times n}$, $b\in\R^d$ and $c\in\R^n$.
Each of \cite{cls19,lsz19,b20,JiangSWZ20} developed their own data structure that maintained one of the following terms
\begin{align}
\text{\cite{cls19}}:&~\phantom{\mR} \mD\mA^\top(\mA\mD^2\mA^\top)^{-1}\mA\mD \label{eq:intro:cls}\\
\text{\cite{lsz19}}:&~ \mR\mD\mA^\top(\mA\mD^2\mA^\top)^{-1}\mA\mD \\
\text{\cite{b20}}:&~\phantom{\mR} \mD\mA^\top(\mA\mD^2\mA^\top)^{-1}\mA\mD h \label{eq:intro:b20}\\
\text{\cite{JiangSWZ20}}:&~ \mR\mD\mA^\top(\mA\mD^2\mA^\top)^{-1}\mA\mD h \label{eq:intro:jswz}
\end{align}
where $\mA$ is the constraint matrix of the linear program,
$\mR$ is a random sketch-matrix, 
$\mD$ is a diagonal matrix that changes over time,
and $h$ is a vector that changes over time as well.
Since the terms differ, each of \cite{cls19,lsz19,b20,JiangSWZ20} had to construct a new data structure.
However, thanks to \Cref{thm:informal_reduction} we now know that it suffices 
to only consider maintaining $\mM^{-1}$ for some matrix $\mM$ that changes over time. 
More accurately, we can unify the data structures for maintaining \eqref{eq:intro:cls}-\eqref{eq:intro:b20}
via the reduction of \Cref{thm:informal_reduction} and a single dynamic matrix inverse algorithm.
This demonstrates Point \ref{point:proofs}.

Further, this also allows to simplify the existing analysis: 
With each increment in length of the above terms \eqref{eq:intro:cls}-\eqref{eq:intro:jswz}, 
the respective data structures developed in \cite{cls19,lsz19,b20,JiangSWZ20} grow more complicated as well.
The most recent result \cite{JiangSWZ20} by Jiang et al.~dedicates the majority of the paper to their new data structure (90 pages). 
By exploiting \Cref{thm:informal_reduction} we can replace the majority (50 pages) of the data structure analysis in \cite{JiangSWZ20} 
by reducing to existing dynamic matrix inverse data structures.
This demonstrates Point \ref{point:reduction}.

\paragraph{Pivoting Based Algorithms}

We show that the time per iteration of pivoting based algorithms
can be reduced from $O(n^2)$ to $O(n^{1.529})$ by using \Cref{thm:informal_reduction} and a data structure by Sankowski \cite{Sankowski04}.

This improves the complexity of the simplex algorithm for solving linear programs.
While theoretic bounds on the complexity of the simplex method are super-polynomial
because of the large number of pivoting steps \cite{KleeM72,Kalai92,HansenZ15},
the simplex algorithm is observed to run in $O(n)$ iterations on practical problem instances \cite{Shamir1987}
or $O(n^3)$ when adding small random noise to the input instance 
\cite{SpielmanT04,DeshpandeS05,KelnerS06,Vershynin09,DadushH18}.
Thus a polynomial speed-up of the time per iteration is non-negligible for those types of instances.
Indeed, analyzing the time per iteration has received attention in the past,
and different trade-offs between time, space, and numeric precision have been obtained
\cite{Bartels68,BartelsG69,Reid82,Goldfarb77a}.
However, none of these results managed to achieve $O(n^{2-\epsilon})$ time per iteration for constant $\epsilon>0$.
Thus our $O(n^{1.529})$ result demonstrates how our technique can be use to speed-up iterative algorithms.

Besides of the simplex algorithm, there also exist other pivoting based algorithms,
some of which have better worst-case bounds on the number of iterations.
For example, there exists a simple $O(n^3)$ time pivoting-based algorithm (with $O(n)$ iterations) for constructing basic solutions.
Beling and Megiddo \cite{BelingM98} present an alternative algorithm that runs in $O(n^{2.529})$ time,
though their new algorithm is also more complicated.
By using \Cref{thm:informal_reduction} we speed up the simple $O(n^3)$ time algorithm to run in $O(n^{2.529})$ time instead. 
This demonstrates Point \ref{point:iterative} because now the simple algorithm runs with the same complexity as the more complicated algorithm from \cite{BelingM98}.

\paragraph{Reproducing and simplifying further results}

To further demonstrate how \Cref{thm:informal_reduction} can be used to speed up simple iterative algorithms,
we show how the QR decomposition of a matrix can be computed in $O(n^\omega)$ time\footnote{%
Here $\omega$ is the matrix exponent, i.e.~multiplying two $n\times n$ matrices can be done in $O(n^\omega)$ operations.}
by speeding up the Gram-Schmidt procedure.
It was already known that QR decompositions can be computed in $O(n^\omega)$ time \cite{ElmrothG00,DemmelDH07},
the point here is that the simple Gram-Schmidt procedure can be sped up from $O(n^3)$ to $O(n^\omega)$ time.

We also show how to reproduce the result on ``online linear systems'' by Storjohann and Yang \cite{StorjohannY15}
via \Cref{thm:informal_reduction} and existing dynamic matrix inverse data structures.
This again shows that many different results can be reproduced by reducing to dynamic matrix inverse,
reducing the overall amount of proofs required in the area.

%% file: preliminaries.tex
\section{Preliminaries}

\paragraph{Misc}
For $n \in \N$ we write $[n]$ for the set $\{1,...,n\}$.
For sets $I \subset [n]$, $J \subset [m]$ and an $n\times m$ matrix $\mM$
we write $\mM_{I,J}$ for the submatrix consisting of rows with index in $I$ and columns with index in $J$.
We write $\mI^{(n)}$ for the $n\times n$ identity matrix and $\zeromat^{(n,m)}$ for an $n\times m$ all zero matrix,
though for simplicity we may also just write $\mI$ if the dimension of the identity matrix is clear from context.
Note that selecting some submatrix $\mM_{I,J}$ of an $n\times m$ matrix can be done via the product
$\mI^{(n)}_{I,[n]} \mM \mI^{(m)}_{[m],J} = \mM_{I,J}$.
We use $\tilde{O}(\cdot)$ to hide polylog$(n)$ factors.
We write $\omega(\cdot,\cdot,\cdot)$ for the matrix exponent, i.e. multiplying an $n^a \times n^b$ by $n^b \times n^c$ matrix requires $O(n^{\omega(a,b,c)})$ operations.
For simplicity we also write $\omega$ as shorthand for $\omega(1,1,1)$.

\paragraph{Matrix formulas}

The terminology of formulas is already known 
in the context of boolean circuits/formulas
which represent boolean functions \cite{Papadimitriou94},
or the context of arithmetic circuits/formulas
which represent functions over groups/rings/fields \cite{BurgisserCS97}.
Note that the set of matrices does not form a group,
because we can only add matrices if the dimensions match,
so we need to adapt the definition of arithmetic formulas a bit
to reflect this requirement.

\begin{definition}
We call a directed tree $f$ (with edges directed towards the root)
a ``$p$-input matrix formula over field $\F$'',
if the nodes of the tree are labeled such that
\begin{itemize}[nosep]
\item Each node is either a \emph{constant-node}, \emph{input-node}, or \emph{gate-node}.

\item Every constant-node $v$ has in-degree zero (i.e. it is a leaf),
is labeled by numbers $n_v, m_v$
and a matrix from $\F^{n_v \times m_v}$.

\item There are $p$ input-nodes.
Each input node $v$ has in-degree zero (i.e. it is a leaf),
is labeled by numbers $n_v, m_v$,
and the symbol $\mM_v$ representing an $n_v \times m_v$ matrix.

\item Every gate-node $v$ is labeled by 
an arithmetic operation $+$ (addition), $\cdot$ (multiplication), or $(\cdot)^{-1}$ (inversion),
and two numbers $n_v$ and $m_v$.

Gate-nodes $w$ labeled with $+$ or $\cdot$ have in-degree $2$. 
Let $w_L,w_R \in V$ be the left and right child of $w$
If the label of $w$ is $+$ then we have $n_v = n_{v_L} = n_{v_R}$ and $m_v = m_{v_L} = m_{v_R}$,
while for label $\cdot$ we have $n_v = n_{v_L}$, $m_v = m_{v_R}$, $m_{v_L} = n_{v_R}$.

Gate-nodes $w$ labeled with $(\cdot)^{-1}$ have in-degree $1$.
Let $w'$ be the child of $w$,
then $n_w=m_w=n_{w'} = m_{w'}$.

\item There is one node with out-degree $0$ (i.e. the root), called \emph{output-node}.
\end{itemize}
We may also write $\mM_1,...,\mM_p$ for the $p$ symbolic matrices of the input-gates.
\end{definition}

Let $f$ be a $p$-input matrix formula,
where each input $M_i$ is of dimension $n_i \times m_i$
and let $v$ be the output-node of the formula.
Then the formula $f$ represents some function $g: {{\times}}_{i=1}^d \F^{n_i \times m_i} \to \F^{n_v\times m_v}$
in the natural way.
For notational simplicity we also write $f$ for the function represented by the formula. 

For matrices $\mA_1,...,\mA_p$ we call $f(\mA_1,...,\mA_p)$ \emph{executable},
if the formula does not try to invert a non-invertible matrix.
Note that \emph{executable} is a property that depends on both $f$ and the input matrices $\mA_1,...,\mA_p$,
i.e.~there exists formulas that are never executable (e.g.~$f = \zeromat^{-1}$)
and there exist formulas that are only executable for some inputs (e.g.~$f(\mM) = \mM^{-1}$).

\paragraph{Block matrices}

An $n \times m$ block matrix $\mB$ is an $n \times m$ matrix 
with an explicit (or implicit via context) partition of it's rows and columns, 
i.e.~$\dot{\bigcup}_{i=1}^s I_i = [n]$, $\dot{\bigcup}_{i=1}^t J_i = [m]$.
Each submatrix $\mB_{I_i,J_j}$ for $1\le i\le s$, $1\le j\le t$ obtained by this partition is called a block of $\mB$.
For example given some $n \times n$ matrix $\mA$, 
the following is a $2n \times 2n$ block matrix with blocks $\mA$, $\mI^{(n)}$, $\mI^{(n)}$, and $\zeromat^{(n,n)}$:
\begin{align*}
\begin{bmatrix}
\mA & \mI^{(n)} \\
\mI^{(n)} & \zeromat^{(n,n)}
\end{bmatrix}.
\end{align*}

\paragraph{Symbolic block matrices}

We define the term ``symbolic block matrix''.
A symbolic block matrix is just a block matrix where an entire block is allowed to be an abstract symbol (i.e.~a variable).
For example the following matrix $\mC$ is a symbolic block matrix
\begin{align*}
\mC = \left[\begin{array}{cccc}
1 & 2 & 3 & 4 \\
5 & 6 & 7 & 8 \\
\cline{1-2}
\multicolumn{2}{|c|}{\multirow{2}{*}{$\mM$}} & 9 & 0\\
\multicolumn{2}{|c|}{} & 1 & 2 \\
\cline{1-2}
\end{array}\right]
\end{align*}
Here $\mM$ is a symbol
and we also identify $\mC$ as the function $\mC: \R^{2 \times 2} \to \R^{4 \times 4}$,
that substitutes the symbol $\mM$ by some given $2 \times 2$ input matrix. 
So for example
\begin{align*}
\mC\left(\left[\begin{array}{cc}
1 & 2 \\
3 & 4
\end{array}\right]\right) = 
\left[\begin{array}{cccc}
1 & 2 & 3 & 4 \\
5 & 6 & 7 & 8 \\
\cline{1-2}
\multicolumn{1}{|c}{1}&\multicolumn{1}{c|}{2} & 9 & 0\\
\multicolumn{1}{|c}{3}&\multicolumn{1}{c|}{4} & 1 & 2 \\
\cline{1-2}
\end{array}\right]
\end{align*}
Formally we define a symbolic block matrix as follows.

\begin{definition}
A symbolic block matrix over $\F$ is a block matrix where each block is either a matrix over $\F$ or a unique symbol, 
i.e.~the same symbol is not allowed in more than one block.

Let $\mC$ be a symbolic block matrix of size $n \times m$ 
with symbols $\mM_1,...,\mM_d$,
where each $\mM_i$ is in a block of size $n_i \times m_i$.
Then we identify $\mC$ also with the function $\mC : \times_{i=1}^d \F^{n_i \times m_i} \to \F^{n \times m}$.
That is, for matrices $\mA_i \in \F^{n_i \times m_i}$ 
the matrix $\mC(\mA_1,...,\mA_d)$ is obtained by substituting each $\mM_i$ by $\mA_i$.
\end{definition}

%% file: matrix.tex
\section{Reducing Formulas to Matrix Inverse}
\label{sec:matrix}

In this section we prove the main tool for obtaining \Cref{thm:informal_reduction}
(and its formal variant \Cref{thm:reduction}).
We want to show that for any $p$-input matrix formula $f$, 
the task of maintaining $f(\mM_1,...,\mM_p)$ under updates to $\mM_i$, 
reduces to the problem of maintaining a submatrix of some inverse matrix $\mN(\mM_1,...,\mM_p)^{-1}$
for some symbolic block matrix $\mN$.
The following theorem shows that such a symbolic block matrix $\mN$ can be constructed for any formula $f$.

\begin{theorem}
\label{thm:main}
Given a $p$-input matrix formula $f$ over field $\F$,
define $n := \sum_{v\in V} n_v + m_v$.

Then there exists a symbolic block matrix $\mN$ of size 
at most $n\times n$,
and sets $I,J \subset [n]$, 
such that for all matrices $\mA_1,...,\mA_p$
for which $f(\mA_1,...,\mA_d)$ is executable,
$(\mN(\mA_1,...,\mA_d)^{-1})_{I,J} = f(\mA_1,...,\mA_d)$.

Constructing $\mN$ from $f$ can be done in $O(n^2)$ time.
\end{theorem}

We want to give an example of \Cref{thm:main}.
Consider the formula $f(\mM, v) = \mM^{-1} v$
for some $n\times n$ matrix $\mM$ and $n$-dimensional vector $v$.
We define the matrix $\mN$ as follows and state its inverse:
\[
\mN = \begin{bmatrix}
\mM & v \\
\zeromat^{(1,n)} & -1
\end{bmatrix}
~~
\mN^{-1} =
\begin{bmatrix}
\mM^{-1} & \mM^{-1}v \\
\zeromat^{(1,n)} & -1
\end{bmatrix}.
\]
Here $\mN$ is a symbolic block matrix where $\mM$ and $v$ are the symbolic blocks.
Further there exists a block of the inverse $\mN^{-1}$ 
(in this case the top-right one) 
that is exactly $f(\mM, v)$.
\Cref{thm:main} claims that such a block matrix exists for every matrix formula $f$.
The proof will use the following observation about the inverse of block matrices:

\begin{lemma}\label{fact}
Given $n \times n$ matrix $\mA$,
$n \times m$ matrix $\mB$,
$m \times n$ matrix $\mC$,
and $m \times m$ matrix $\mD$,
define
\begin{align}
\mG(\mA,\mB,\mC,\mD) := \begin{bmatrix}
\mA & \mB \\
\mC & \mD
\end{bmatrix}.
\end{align}
If $\mA$ and $(\mD-\mC\mA^{-1}\mB)$ are invertible then
\begin{align}
\mG(\mA,\mB,\mC,\mD)
^{-1}
=
\begin{bmatrix}
\mA^{-1} + \mA^{-1}\mB(\mD-\mC\mA^{-1}\mB)^{-1}\mC\mA^{-1} & -\mA^{-1}\mB(\mD-\mC\mA^{-1}\mB)^{-1} \\
-(\mD-\mC\mA^{-1}\mB)^{-1}\mC\mA^{-1} & (\mD-\mC\mA^{-1}\mB)^{-1}
\end{bmatrix}.
\end{align}
\end{lemma}

We omit the proof of \Cref{fact} here as the statement can be verified trivially
by just multiplying the block matrix with the claimed inverse 
and checking that the result is the identity matrix.

The idea for proving \Cref{thm:main} is to construct a symbolic block matrix for each matrix operation.
Thus for each matrix operation we have some gadget,
and then we combine these gadgets recursively to construct a matrix $\mN$ for any formula $f$.
We give the following constructive proof on how to obtain $\mN$:

\begin{proof}[Proof of \Cref{thm:main}]
We prove \Cref{thm:main} via induction over the number of gates (i.e.~nodes) in formula $f$.
\paragraph{One gate}
If $f$ consists of a single gate $v$,
let $\mM_v$ be the $n_v \times m_v$ (possibly symbolic) matrix that $v$ is labeled by.
We define the $(n_v+m_v) \times (n_v+m_v)$ matrix
\begin{align*}
\mN = \left[\begin{matrix}
\mI^{(n_v)} & \mM_v \\
\zeromat^{(m_v,n_v)} & -\mI^{(m_v)}
\end{matrix}\right].
\end{align*}
Here we have $\mN^{-1} = \mN$, so for $I = \{1,...,n_v\}$, $J = \{n_v+1,...,n_v+m_v\}$
we have $(\mN^{-1})_{I,J} = \mM_v$.

\paragraph{Inductive Step}

Assume \Cref{thm:main} holds for formulas with at most $K$ gates.
Let $f$ be a formula with $K+1$ gates
and let $v$ be the root (i.e.~the output-gate).

We now have several cases, depending on which operation $v$ is labeled with.

\paragraph{Inversion}
Let $w$ be the child of $v$
and let $f'$ be the formula represented by the subtree rooted at $w$,
so $f(\mM_1,...,\mM_d) = f'(\mM_1,...,\mM_d)^{-1}$.

By induction hypothesis we have a symbolic block matrix $\mN'$
and sets $I',J'$ with 
\[(\mN'(\mM_1,...,\mM_d)^{-1})_{I',J'} = f'(\mM_1,...,\mM_d).\]
Let $n_{N'}$ be the dimension of $\mN'$, 
then we define $\mN = \mG(\mN', -\mI^{(n_{N'})}_{[n_{N'}],J'}, \mI^{(n_{N'})}_{I',[n_{N'}]}, \zeromat^{(n_w,n_w)})$
for $\mG$ as in \Cref{fact}.
Then by \Cref{fact} the bottom right block of $\mN^{-1}$ is
\[
(\zeromat^{(n_w,n_w)} + \mI^{(n_{N'})}_{I',[n_{N'}]} \mN'^{-1} \mI^{(n_{N'})}_{[n_{N'}],J'})^{-1}
= ((\mN'^{-1})_{I',J'})^{-1}.
\]
Thus for all $\mM_1,...,\mM_d$ 
we have that the bottom right block of $\mN(\mM_1,...,\mM_d)^{-1}$
is exactly $(f'(\mM_1,...,\mM_d))^{-1} = f(\mM_1,...,\mM_d)$,
assuming the inverse exists (i.e.~$f(\mM_1,...,\mM_d)$ is executable).

The number of rows and columns in $\mN$ is 
$(\sum_{x \in V \setminus \{w\}} n_x + m_x) + n_w \le \sum_{x \in V} n_x + m_x$,
where $(\sum_{x \in V \setminus \{w\}} n_x + m_x)$ is the size of $\mN$ by induction hypothesis.

\paragraph{Addition and Subtraction}

Let $w_L,w_R$ be the children of $v$
and let $f_L,f_R$ be the functions represented by the subtrees rooted at $w_L$ and $w_R$ respectively,
so $f(\mM_1,...,\mM_d) = f_L(\mM_1,...,\mM_k) + f_R(\mM_{k+1},...,\mM_d)$ for some $k$.

By induction hypothesis we have symbolic block matrices $\mL, \mR$ 
and respective sets $I_L,J_L$ and $I_R,J_R$ 
with $(\mL(\mM_1,...,\mM_k)^{-1})_{I_L,J_L} = f_L(\mM_1,...,\mM_k)$
and $(\mR(\mM_{k+1},...,\mM_d)^{-1})_{I_R,J_R} = f_R(\mM_{k+1},...,\mM_d)$.
Let $n_L, n_R$ be the dimension of $\mL$ and $\mR$ respectively,
i.e.~$\mL$ is a $n_L \times n_L$ matrix.

Note that $|I_L| = |I_R| = n_w$ and $|J_L|=|J_R|=m_w$
so we can define the following $(n_L+n_R+m_{w} + n_w) \times (n_L+n_R+m_{w} + n_w)$ matrix
\begin{align*}
\mN := \left[\begin{array}{cc|cc}
\mL & \zeromat^{(n_L,n_R)} & \mI^{(n_L)}_{[n_L],J_L} & \zeromat^{(n_L,n_w)}\\
\zeromat^{(n_R,n_L)} & \mR & \mI^{(n_R)}_{[n_R],J_R} & \zeromat^{(n_R,n_w)}\\
\hline
\zeromat^{(m_w,n_L)} & \zeromat^{(m_w,n_R)} & -\mI^{(m_w)} & \zeromat^{(m_w,n_w)}\\
\mI^{(n_L)}_{I_L,[n_L]} & \mI^{(n_R)}_{I_R,[n_R]} & \zeromat^{(n_w,m_w)} & -\mI^{(n_w)}
\end{array}\right].
\end{align*}
The inverse of that matrix is given by
\begin{align*}
\mN^{-1}=
\left[\begin{array}{cc|cc}
\mL^{-1} & \zeromat^{(n_L,n_R)} & (\mL^{-1})_{[n_L],J_L} & \zeromat^{(n_L,n_w)}\\
\zeromat^{(n_R,n_L)} & \mR^{-1} & (\mR^{-1})_{[n_R],J_R} & \zeromat^{(n_R,n_w)}\\
\hline
\zeromat^{(m_w,n_L)} & \zeromat^{(m_w,n_R)} & -\mI^{(m_w)} & \zeromat^{(m_w,n_w)}\\
(\mL^{-1})_{I_L,[n_L]} & (\mR^{-1})_{I_R,[n_R]} & (\mL^{-1})_{I_L,J_L}+(\mR^{-1})_{I_R,J_R} & -\mI^{(n_w)}
\end{array}\right].
\end{align*}
This can easily be verified by just multiplying the two matrices, as the result is a large identity matrix.
Note that the product of these two matrices could be computed by interpreting them as $4\times 4$ matrices 
and multiplying the blocks in the same way as one would multiply entries of a matrix. 
This way one can quickly verify that this is indeed the inverse of $\mN$.

Now consider the block of the inverse which is at the bottom, second from the right. That block satisfies
\begin{align*}
&~
(\mL(\mM_1,...,\mM_k)^{-1})_{I_L,J_L}+(\mR(\mM_{k+1},...,\mM_d)^{-1})_{I_R,J_R} \\
=&~ 
f_L(\mM_1,...,\mM_k) + f_R(\mM_{k+1},...,\mM_d) \\
=&~
f(\mM_1,...,\mM_d).
\end{align*}
Note that subtraction can be handled in the same way, 
by simply using $-\mI^{(n_R)}_{[n_R],J_R}$ for the block second from the top, third from the left, in matrix $\mN$.

\paragraph{Multiplication}

Let $w_L,w_R$ be the children of $v$
and let $f_L,f_R$ be the formulas represented by the subtree rooted as $w_L$ and $w_R$ respectively,
so $f(\mM_1,...,\mM_d) = f_L(\mM_1,...,\mM_k) \cdot f_R(\mM_{k+1},...,\mM_d)$ for some $k$.

By induction hypothesis we have symbolic block matrices $\mL, \mR$ 
and respective sets $I_L,J_L$ and $I_R,J_R$ 
with $(\mL(\mM_1,...,\mM_k)^{-1})_{I_L,J_L} = f_L(\mM_1,...,\mM_k)$
and $(\mR(\mM_{k+1},...,\mM_d)^{-1})_{I_R,J_R} = f_R(\mM_{k+1},...,\mM_d)$.
Let $n_L, n_R$ be the dimension of $\mL$ and $\mR$ respectively,
i.e.~$\mL$ is a $n_L \times n_L$ matrix.

We define the $(n_L + n_R) \times (n_L + n_R)$ matrix
$\mN = \mG(\mL, -\mI^{(n_L)}_{[n_L],J_L}\mI^{(n_R)}_{I_R,[n_R]}, \zeromat^{(n_R,n_L)}, \mR)$
for $\mG$ as in \Cref{fact} so the top right block of $\mN^{-1}$ is
\begin{align*}
&~
\mL^{-1} \mI^{(n_L)}_{[n_L],J_L}\mI^{(n_R)}_{I_R,[n_R]} (\mR - \zeromat^{(n_R,n_L)} \mL^{-1} \mI^{(n_L)}_{[n_L],J_L}\mI^{(n_R)}_{I_R,[n_R]} )^{-1} \\
=&~
\mL^{-1} \mI^{(n_L)}_{[n_L],J_L}\mI^{(n_R)}_{I_R,[n_R]} \mR^{-1} \\
=&~
(\mL^{-1})_{[n_L],J_L} (\mR^{-1})_{I_R,[n_R]}.
\end{align*}
So picking rows $I_L$ and columns $J_R$ of that block, 
we obtain 
\[(\mL(\mM_1,...,\mM_k)^{-1})_{I_L,J_L} (\mR(\mM_{k+1},...,\mM_d)^{-1})_{I_R,J_R} = f(\mM_1,...,\mM_d).\]
\paragraph{Construction time}
Constructing the matrix $\mN$ can be implemented to take $O(n^2)$ time, 
if one does not copy the matrices that represent the children in each inductive step 
and one instead uses pointers to reference them,
i.e~during construction one uses a recursive tree structure to represent the matrix.
Only at the very end, one constructs the matrix $\mN$ explicitly by following the pointers.
\end{proof}

%% file: application.tex
\section{Applications}
\label{sec:applications}

In this section we show how to use the construction of \Cref{thm:main}
together with existing dynamic matrix inverse algorithms
to simplify and speed up some iterative algorithms.
In the introduction we already gave some examples of how our result can be used,
here we prove them more formally.

We start by formally reducing dynamic matrix formula to dynamic matrix inverse
in \Cref{sec:formula} and we list the data structures we obtain from this reduction.
The applications for
pivoting-based algorithms (e.g.~the simplex algorithm) are stated in \Cref{sec:pivoting}
and in \Cref{sec:ipm} we discuss how to use dynamic matrix inverse algorithms for interior point based linear program solvers.
To further demonstrate how these data structures can be used 
to obtain very fast iterative algorithms in a blackbox way, 
we demonstrate in \Cref{sec:onlinematrix} 
how to reproduce the online matrix inversion result 
and in \Cref{sec:qr} we show how to speed up the Gram-Schmidt procedure to run in $O(n^\omega)$ time.
At last, we discuss some insights regarding fine-grained complexity theory 
obtained from \Cref{thm:main} in \Cref{sec:finegrained}.

\subsection{Dynamic Matrix Formula}
\label{sec:formula}

The formal variant of \Cref{thm:informal_reduction} can be stated as follows:

\begin{theorem}[Formal variant of \Cref{thm:informal_reduction}]\label{thm:reduction}
The dynamic matrix formula problem is a data structure problem where
one is given a $p$-input matrix formula $f$ over field $\F$
and corresponding input matrices $\mM_1,...,\mM_p$.
After some initial preprocessing, 
the data structure allows for updates (changes to $\mM_1,...,\mM_p$)
and queries (retrieving entries of $f(\mM_1,...,\mM_p)$), 
assuming $f(\mM_1,...,\mM_p)$ is executable throughout all updates.

Define $n := \sum_{v\in V} n_v + m_v$ where $V$ are the gates of formula $f$
and assume there exists a dynamic matrix inverse data structure $\cA$
that, after some initial preprocessing, 
maintains the inverse of a non-singular matrix $\mA \in \F^{n\times n}$
under updates to $\mA$ and allows for querying entries of $\mA^{-1}$.

Then there exists a dynamic matrix formula data structure $\cF$ for maintaining $f$.
If $\cA$ supports changing any entry, column, or row of $\mA$,
then $\cF$ supports the same types of changes to the input matrices $\mM_1,...,\mM_p$.
If $\cA$ supports querying entries, rows, or columns of $\mA^{-1}$,
then $\cF$ supports the same types of queries to $f(\mM_1,...,\mM_p)$.
The preprocessing, update, and query complexities of $\cF$ are the same as the ones of $\cA$.
\end{theorem}

Note that, while the formula $f$ is part of the input,
only the matrices $\mM_1,...,\mM_p$ are allowed to change.

\begin{proof}
Given formula $f$ we apply \Cref{thm:main} to construct some $n\times n$ symbolic block matrix $\mN$ and sets $I,J\subset[n]$ such that $f(\mM_1,...,\mM_p) = (\mN(\mM_1,...,\mM_p)^{-1})_{I,J}$.\footnote{%
$\mN$ might be of smaller dimension than $n$, but that would only improve the complexities.}
For the given input matrices $\mM_1,...,\mM_p$ we then initialize the dynamic matrix inverse data structure $\cA$ on
$\mN(\mM_1,...,\mM_p)$.
Since each matrix $\mM_i$ is just a block of $\mN(\mM_1,...,\mM_p)$, 
changing an entry, column, or row of some $\mM_i$
corresponds to changing some entry, column or row of $\mN(\mM_1,...,\mM_p)$ respectively.
Further, querying an entry, row or column of $f(\mM_1,...,\mM_p) = (\mN(\mM_1,...,\mM_p)^{-1})_{I,J}$
can be done by querying an entry, row or column of $\mA(\mM_1,...,\mM_p)^{-1}$.
\end{proof}

Using \Cref{thm:reduction} and existing data structures for dynamic matrix inverse
we obtain the following data structures for dynamic matrix formula.

\begin{corollary}\label{thm:formula}
There exist the following dynamic matrix formula data structures with different types of update and query.
Given a $p$-input formula $f$ and initial input matrices $\mM_1,...,\mM_p$,
define $n := \sum_{v\in V} n_v + m_v$ where $V$ are the gates of formula $f$.
\begin{itemize}[nosep]
\item Changing any entry of any $\mM_1,...,\mM_p$ 
and querying any entry of $f(\mM_1,...,\mM_p)$ in $O(n^{1.407})$ field operations.
\item Changing any entry of any $\mM_1,...,\mM_p$ in $O(n^{1.529})$ field operations 
and querying any entry of $f(\mM_1,...,\mM_p)$ in $O(n^{0.529})$ field operations.
\item Changing any column of $\mM$ 
and querying any row of $f(\mM_1,...,\mM_p)$ in $O(n^{1.529})$ field operations.
\item Changing any column of $\mM$ 
and querying any row of $f(\mM_1,...,\mM_p)$ in $O(n^{\omega-1})$ field operations,
if the position of all updates and queries (i.e.~column and row indices) is a fixed sequence 
given to the data structure during the very first update.
\end{itemize}
All these data structures require $O(n^\omega)$ field operations for their preprocessing.
For the data structures with row updates and column queries 
there also exist variants with column updates and row queries.
\end{corollary}

For the complexity we remark that any formula with $k$ matrix operations
and matrices whose dimension is at most $m\times m$
we have $n = O(km)$ in \Cref{thm:formula}.
For most applications the number of operations is usually $O(1)$,
so the complexity is bounded by the largest dimension of any of the matrices.

\begin{proof}
There exist dynamic matrix inverse data structures with the following update and query operations:
\begin{itemize}[nosep]
\item[\cite{BrandNS19}] Changing any entry of $\mM$ and querying any entry of $\mM^{-1}$ in $O(n^{1.407})$ field operations. 
\item[\cite{Sankowski04}] Changing any entry of $\mM$ in $O(n^{1.529})$ time and querying any entry of $\mM^{-1}$ in $O(n^{0.529})$ field operations.
\item[\cite{BrandNS19}] Changing any column of $\mM$ and querying any row of $\mM^{-1}$ in $O(n^{1.529})$ field operations.
\item[\cite{BrandNS19}] Changing any column of $\mM$ and querying any row of $\mM^{-1}$ in $O(n^{\omega-1})$ field operations,
if the position of all updates and queries (i.e.~column and row indices) is a fixed sequence 
given to the data structure during the very first update.
\end{itemize}
All these data structures require $O(n^\omega)$ field operations for their preprocessing.
\Cref{thm:formula} is now obtained via the reduction of \Cref{thm:reduction}.

At last, note that for any non-singular matrix $\mA$ we have $(\mA^{-1})^\top = (\mA^\top)^{-1}$,
so a dynamic matrix inverse algorithm with column updates and row queries
can also be used as an algorithm with row updates and column queries by receiving the transpose as input.
Thus we also have data structures for row updates to $\mM_i$ and column queries to $f(\mM_1,...,\mM_p)$.
\end{proof}

\subsection{Pivoting Algorithms}
\label{sec:pivoting}

Here we show how to use our reduction for pivoting based algorithms.
As demonstration we will speed up the simplex algorithm
and simplify the algorithm for constructing basic solutions
by Beling and Megiddo \cite{BelingM98}.

\paragraph{Simplex Algorithm}

Given a linear program of the form
$$
\max_{\mA x = b,~ x \ge 0} c^\top x
$$
for a $d \times n$ matrix $\mA$ of rank $d$,
the simplex algorithm maintains a base $B$ of $\mA$,
i.e.~a subset of columns of $\mA$.
In each iteration one column of $\mA$ is removed from $B$ while another column is added.
Let's say the base $B \subset [n]$ is represented as a set of $d$ indices
and let $\mA_B$ be the submatrix of $\mA$ consisting of the base vectors,
then the simplex algorithm (and other pivoting based algorithms) must repeatedly
solve a linear system in $\mA_B^{-1}$.
More accurately, the simplex algorithm maintains a tableau of the following form
\begin{align*}
\left[
\begin{array}{ccc|c}
\omA & \mI & \zerovec & \ob \\
\oc & \zerovec^\top & 1 & \of
\end{array}
\right]
\end{align*}
where $\omA = \mA_B^{-1} \mA_N$,
$\ob = \mA_B^{-1} b$,
$\oc = c^\top_N - c^\top_B \mA_B^{-1} \mA_N$,
$\of = c^\top_B \mA_B^{-1} b$.
Here $N = [n] \setminus B$ are the non-basis indices.

After computing this tableau,
the simplex algorithm picks a column-pivot $j \in N$ with $\oc_j > 0$
then picks a row-pivot $i \in B$ with $\omA_{i,j} > 0$ which minimizes $\lambda = \min_i \ob_i / \omA_{i,j}$.
Then $i$ is removed from $B$ and $j$ is added to it.

Thus to implement one such step of the simplex algorithm efficiently,
we need a data structure that can maintain
all entries of $\oc$ and $\ob$, 
and the data structure must be also able to answer queries for any column of $\omA$.
For that note that we can write $\omA$ as
\begin{align}
\omA = \mA_B^{-1} \mA_N = (\mA \mI_{[n],B})^{-1} \mA \mI_{[n],N}
\label{eq:simplex}
\end{align}
and when changing an element in $B$ and $N$,
we just need to perform two element updates to $\mI_{[n],B}$ and $\mI_{[n],N}$
in order to obtain the new $\omA$.
We use \Cref{thm:formula} (more accurately the $O(n^{1.529})$ update and $O(n^{0.529})$ query variant)
to maintain the right-most formula of \eqref{eq:simplex}
for the aforementioned element updates.
Whenever a column of $\omA$ is required, we simply query the $n$ entries one by one in $O(n^{1.529})$ total time.
A similar formula as \eqref{eq:simplex} can also be constructed for $\ob$ and $\oc$,
thus we require only $O(n^{1.529})$ time per iteration.
For comparison, naively solving the linear system in $\mA_B^{-1}$ in each iteration takes $O(n^\omega)$ time,
which in most variants of the simplex method is sped up to $O(n^2)$ time 
by updating $\mA_B^{-1}$ (e.g.~via Gaussian elimination) whenever $B$ changes \cite{Bartels68}.

\paragraph{Basic Solutions}

Note that the idea to solve linear systems in $\mA_B$ via the formula $(\mA \mI_{[n],B})^{-1}$ 
can be applied to other pivoting algorithms as well.
Beling and Megiddo \cite{BelingM98} discussed how to construct an optimal basic solution, given some optimal solution for a linear program.
The naive algorithm requires $O(n^3)$ time ($O(n^2)$ time per iteration over $O(n)$ iterations)
and repeatedly computes a column of $\omA = \mA_B^{-1}\mA_N$ 
while replacing an entry in $B$ in each iteration,
just like the simplex method.
Beling and Megiddo modify this algorithm so that it can exploit fast matrix multiplication,
which improves complexity to $O(n^{2.529})$ time \cite{BelingM98}.
However, we observe that their modification is not required,
instead one can simply maintain $\omA$ via \Cref{thm:formula} in $O(n^{1.529})$ time per iteration
as done in the previous example for the simplex algorithm.
This directly yields a total time of $O(n^{2.529})$, 
matching the result of Beling and Megiddo 
but without changing the iterative algorithm itself.
Instead we just plug in a data structure to speed up the simple iterative algorithm.

\subsection{Interior Point Method}
\label{sec:ipm}

Consider a linear program of the form
$$
\min_{\mA x = b, x \ge 0} c^\top x,
$$
where $\mA$ is a $d\times n$ matrix and $n\ge d$.
Recent work on linear program solvers reduced solving linear programs to maintaining the following expressions
\begin{align}
\text{\cite{cls19}}:&~\phantom{\mR} \mD\mA^\top(\mA\mD^2\mA^\top)^{-1}\mA\mD \label{eq:ipm:cls}\\
\text{\cite{lsz19}}:&~ \mR\mD\mA^\top(\mA\mD^2\mA^\top)^{-1}\mA\mD \label{eq:ipm:lsz}\\
\text{\cite{b20}}:&~\phantom{\mR} \mD\mA^\top(\mA\mD^2\mA^\top)^{-1}\mA\mD h \label{eq:ipm:bra}\\
\text{\cite{JiangSWZ20}}:&~ \mR\mD\mA^\top(\mA\mD^2\mA^\top)^{-1}\mA\mD h \label{eq:ipm:jswz}
\end{align}
where $\mA$ is the constraint matrix of the linear program,
$\mR$ is a random sketch-matrix, 
$\mD$ is a diagonal matrix that changes over time,
and $h$ is a vector that changes over time as well.

Each of \eqref{eq:ipm:cls}-\eqref{eq:ipm:jswz} constructed their own data structure for maintaining these expressions.
Here we show that the first three can be reproduced via a single unifying data structure 
and the last once can be reduced to an existing dynamic matrix inverse algorithm.
We start with the expression \eqref{eq:ipm:bra} as used by v.d.Brand~\cite{b20}, 
as that one is easiest to explain how to maintain it via our reduction.

\paragraph{Reproducing \cite{b20}:}

The data structure task in \cite{b20} can be stated as follows
\begin{task}\label{task:lp}
Maintain \eqref{eq:ipm:bra} under updates to $\mD$ and $h$.
There are $\tilde{O}(\sqrt{n})$ iterations in total.
The number of entries of $\mD$ and $h$ that change in any iteration is a power of $2$,
and changing $2^j$ entries occurs at most once every $\tilde{O}(2^{j/2})$ iterations for $j=0,...,\log n$.
\end{task}
We will prove the following upper bound for \Cref{task:lp} which matches the complexity proven in \cite{b20}.
\begin{theorem}\label{thm:lp}
We can solve \Cref{task:lp} in $\tilde{O}(n^{1+x} + n^{\omega(1,1,x)-x/2} + n^{\omega-1/2})$
amortized time per iteration for any $0 \le x \le \alpha$.\footnote{%
Here $\alpha$ is the largest number such that $\omega(1,1,\alpha)=2$.}
\end{theorem}
By the guarantees of \Cref{task:lp} we know the total number of entries 
that must be changed in $\mD$ and $h$ is bounded by $\tilde{O}(n)$.
Thus using the $O(n^{1.529})$-time element update, $O(n^{0.529})$-time element query variant
of \Cref{thm:formula}, we would need $\tilde{O}(n^{2.529})$ time in total over all $\tilde{O}(\sqrt{n})$ iterations.
This can be improved further because the dynamic matrix inverse data structure we used here
is optimized as to minimize the update time \emph{per changed entry}.
However, we can improve the bound by analyzing the update time \emph{per iteration}.
For that we will use the following dynamic matrix inverse data structure
(proven in \Cref{sec:appendix}).
\begin{restatable}{lemma}{exactComplexity}
\label{lem:exact_complexity}
Given a non-singular $n\times n$ matrix $\mM$,
the data structure preprocesses the matrix in $O(n^\omega)$ time.
Afterwards the data structure supports the following operations
\begin{itemize}[nosep]
\item \textsc{Update}$(\Delta)$: 
The data structure stores the given $n \times n$ matrix $\Delta$.
Assume $\Delta$ is given in a sparse representation and changes at most $n^a$ entries in $n^b$ columns,
then the update time is 
$O(n^{a+b} + n^{b\cdot\omega}).$
\item \textsc{Query}$(I,j)$: 
For given $I \subset [n]$, $j \in [n]$, returns $((\mM+\Delta)^{-1})_{I,j}$
in $O(|I|\cdot n^a + n^{2b})$ time.
\item \textsc{Reset}: 
Sets $\mM \leftarrow \mM + \Delta$ 
and then $\Delta \leftarrow \zeromat^{(n,n)}$ 
in $O(n^{\omega(1,1,b)})$ time.
\end{itemize}
Note that the changes given to \textsc{Update} are temporary 
and will be overwritten by the next call to \textsc{Update},
unless \textsc{Reset} is called.
\end{restatable}

\begin{proof}[Proof of \Cref{thm:lp}]
We want to maintain the value of the formula \eqref{eq:ipm:bra}
while $\mD$ and $h$ change over time.
Note that each matrix (and the vector $h$) has their dimension bounded by $n$,
so applying \Cref{thm:main} results in a matrix $\mN$ of size $O(n)\times O(n)$.
We now maintain $\mN$ under element updates and in each iteration via \Cref{lem:exact_complexity}.

Let $U$ be the set of all entry changes 
that we had to perform to $\mD$ and $h$ 
since the last call to \textsc{Reset},
i.e.~$U$ grows in each iteration by the new updates we must perform.
In each iteration, we call \textsc{Update} to perform all $U$ updates at once.
If $U$ changes at least $n^x$ entries for some parameter $0\le x \le 1$,
then we also perform a call to \textsc{Reset}
and the set $U$ becomes an empty set again.
At the end of each iteration, 
perform a call to \textsc{Query} to obtain all $n$ entries of \eqref{eq:ipm:bra}.

Now let us bound the complexity. 
By \Cref{task:lp} we know that between any for two iterations that changed $2^j$ entries,
there must have been at least $\tilde{\Omega}(2^{j/2})$ iterations.
Thus for set $U$ to grow to some size $2^j$,
there must have been at least $\tilde{\Omega}(2^{j/2}/2)$ iterations
and the amortized cost of \textsc{Reset} per iteration is
\begin{align}
\tilde{O}\left(
\sum_{j=x\log n}^{\log n} n^{\omega(1,1,j / \log n)- j/(2 \log n)}
\right)
\label{eq:amortized}
\end{align}
which by convexity of $\omega(1,1,\cdot)$ can be bounded by
$\tilde{O}(n^{\omega(1,1,x)-x/2} + n^{\omega-1/2})$.\\
The time for a call to \textsc{Update} is bounded by
\begin{align*}
\tilde{O}(n^{a+b} + n^{b\cdot\omega})
=
\tilde{O}(n^{1+x} + n^{x\cdot\omega}),
\end{align*}
where $a \le 1$ and $b\le x$, 
because for $b \ge x$ the time for that iteration is dominated by the call to \textsc{Reset}
whose impact we already bounded.

The cost per \textsc{Query} is
$\tilde{O}(n^{1+x} + n^{2x}) = \tilde{O}(n^{1+x})$
as we query $n$ entries in a single column.

Combining the cost per iteration of \textsc{Update} and \textsc{Query}
with the amortized cost of \textsc{Reset} results in an amortized cost per iteration of
\[
\tilde{O}(n^{1+x} + n^{x\omega} + n^{\omega(1,1,x)-x/2} + n^{\omega-1/2}).
\]
In \cite{b20} it was assumed that $x \le \alpha$
in which case the above complexity simplifies to
\[
\tilde{O}(n^{1+x} + n^{2-x/2} + n^{\omega-1/2})
\]
which is exactly the complexity stated in \cite{b20}.
\end{proof}

\paragraph{Reproducing \cite{cls19}:}

The data structure of Cohen et al.~\cite{cls19} can be reproduced in a similar way
as done before for \cite{b20}.
The task to be solved is a bit different:
\begin{task}\label{task:lp_R}
Maintain \eqref{eq:ipm:cls} \emph{implicitly} under updates to $\mD$.
Once per iteration we are given an $n$-dimensional vector $h$
with $\tilde{O}(\sqrt{n})$ non-zero entries
and must return \eqref{eq:ipm:bra}.
\end{task}

Note that we can interpret \Cref{task:lp_R} as a variant of \Cref{task:lp}
where $h$ changes in exactly $\tilde{O}(\sqrt{n})$ entries in every iteration.
Thus we can use the same data structure as in the proof of \Cref{thm:lp}.
Unlike \Cref{task:lp}, however, \Cref{task:lp_R} does not explicitly state how
many entries change in $\mD$. 
Thus we only reproduce the construction of the data structure from \cite{cls19}, 
but not their amortized complexity analysis.
Though since we reproduce their data structure, 
one can apply their amortized complexity analysis.

We use the same data structure as in \Cref{thm:lp},
but now we perform a call to \textsc{Reset} if at least $n^x$
\emph{columns} are changed, instead of entries.
Then every iteration with at most $n^x$ changed columns
takes time 
\begin{align*}
&~
\tilde{O}(\underbrace{(n^{0.5}+n^x)n^x + n^{x\omega}}_{\textsc{Update}} + \underbrace{n(n^{0.5} + n^x) + n^{2x}}_{\textsc{Query}}) \\
=&~
\tilde{O}(n^{1.5} + n^{1+x} + n^{x\cdot\omega})
\end{align*}
to perform the call to \textsc{Update} and \textsc{Query},
and if some $n^y > n^x$ columns of $\mD$ need to change,
then the iteration takes $O(n^{\omega(1,1,y)})$ time 
based on the complexity of \textsc{Reset}.
This matches exactly the complexities of Cohen et al.~\cite[Lemma 5.4, Lemma 5.5]{cls19}
so we managed to reproduce their data structure.
In \cite{cls19} it was shown that for the choice $x \le \alpha$
the expected amortized time per iteration can then be bounded by
\[
\tilde{O}(n^{1.5}+n^{1+x}+n^{\omega-1/2}+n^{2-x/2}).
\]

\paragraph{Reproducing \cite{lsz19}:}

Lee et al.~\cite{lsz19} constructed a data structure for maintaining \eqref{eq:ipm:lsz}
in an implicit way, i.e.~in each iteration a \emph{dense} $n$-dimensional vector $h$ is given
and \eqref{eq:ipm:jswz} must be returned.
Here $\mR$ is an $\tilde{O}(n)\times n$ matrix
and in each iteration one must only return some $\tilde{O}(\sqrt{n})$ entries of \eqref{eq:ipm:jswz}
instead of the entire vector.
Thus the task is somewhat similar to \Cref{task:lp_R} (i.e.~the previous case)
except for the additional matrix $\mR$ multiplied on the left, the density of vector $h$, and the size of the output.

Because of the density of $h$, the cost of \textsc{Update} becomes a bit slower compared to the previous case.
The new complexity of \textsc{Update} is
$
O(n^{1+x}+n^{x\cdot\omega}).
$
The cost of \textsc{Query} is
$
\tilde{O}(n^{1.5} + n^{2x})
$
because we now query $\tilde{O}(\sqrt{n})$ entries,
but the number of changed entries is bounded by $n$.
So overall the cost of \textsc{Update} and \textsc{Query} is again $O(n^{1+x}+n^{1.5}+n^{x\omega})$
as in the previous case of \cite{cls19} (\Cref{task:lp_R}).
The cost of \textsc{Reset} is also the same as before.
This matches the complexities of \cite{lsz19}[Lemma B.5-B.10]
so we managed to reproduce their data structure.
In \cite{lsz19} it was shown that for the choice $x \le \alpha$
the expected amortized time per iteration
is $\tilde{O}(
n^{1+x} + n^{\omega(1,1,x)-x/2} + n^{\omega-1/2} + n^{1.5}
).$

\paragraph{Reproducing \cite{JiangSWZ20}:}

If one assumes $\omega=2$, 
then all previous amortized complexities are $\tilde{O}(n^{1+2/3})$ per iteration,
or $\tilde{O}(n^{2+1/6})$ total time for solving the linear program 
because there are $\tilde{O}(\sqrt{n})$ iterations in total.
It is an open problem whether this can be reduced to $\tilde{O}(n^{1.5})$ amortized time per iteration,
or equivalently $\tilde{O}(n^2)$ total time (assuming $\omega=2$),
which would match the lower bounds.

A breakthrough was made by Jiang et al.~\cite{JiangSWZ20} 
who showed that the total time complexity can be reduced to $\tilde{O}(n^{2+1/18})$
by combining the approach of v.d.Brand \cite{b20} and Lee et al.~\cite{lsz19},
i.e.~showing it suffices to maintain \eqref{eq:ipm:jswz} while both $\mD$ and $h$ change over time
and one must only return $\tilde{O}(\sqrt{n})$ entries of \eqref{eq:ipm:jswz} per iteration.
The majority (over 90 pages) of \cite{JiangSWZ20} is dedicated to constructing and analyzing 
the required data structure that can efficiently maintain \eqref{eq:ipm:jswz}.

Here we show how to obtain such a data structure 
via a short reduction to dynamic matrix inverse algorithms.
We use the algorithm of \Cref{thm:formula} with $O(n^{1.407})$ update time per changed entry,
which actually has the following complexity:\footnote{%
The statement of \cite[Theorem 4.3, Lemma 4.9]{BrandNS19} 
does not explicitly mention the reset-functions
as in \Cref{lem:elementInverse}.
Instead, it states that changing exactly $n^c$ entries in every iteration
takes time
$O(n^{\omega(a,b,c)} + n^{\omega(1,a,b)-b+c} + n^{\omega(1,1,a)-a+c})$ 
per iteration.
This is because the algorithm internally calls \textsc{ResetMiddleLayer}
every $n^{b-c}$ iterations and \textsc{ResetTopLayer} every $n^{a-c}$ iterations.
For $c=0$, balancing the parameters yields the previously stated $O(n^{1.407})$-time bound.}

\begin{lemma}[{\cite[Theorem 4.3, Lemma 4.9]{BrandNS19}}]\label{lem:elementInverse}
Given an $n\times n$ matrix $\mM$,
the data structure preprocesses the matrix in $O(n^\omega)$ time.
Afterward the data structure supports the following operations
\begin{itemize}
\item \textsc{Update}: Change any set of entries of $\mM$.
\item \textsc{Query}: Return any set of entries of some column of $\mM^{-1}$.
\item \textsc{ResetTopLayer} and \textsc{ResetMiddleLayer}: These methods allows to optimize the complexity (see complexities below).
\end{itemize}
Assume at most $n^a$ entries have been changed in $\mM$ since the last call to \textsc{ResetTopLayer}
and at most $n^b$ (for $b \le a$) entries have been changed since the last call to \textsc{ResetMiddleLayer}.
Then the complexity of \textsc{Update} when changing $n^c$ many entries is $O(n^{\omega(a,b,c)})$.
The complexity of \textsc{Query} when asking for $n^c$ entries of the same column of $\mM^{-1}$
is $O(n^{\omega(a,b,c)})$.
The complexity of \textsc{ResetTopLayer} is $O(n^{\omega(1,1,a)})$
and the complexity of \textsc{ResetMiddleLayer} is $O(n^{\omega(1,a,b)})$.
\end{lemma}

Note that unlike \Cref{lem:exact_complexity}, the updates performed in \Cref{lem:elementInverse}
via \textsc{Update} persist
and are not replaced by the next call to \textsc{Update}.

Via the reduction of \Cref{thm:formula},
the task of maintaining \eqref{eq:ipm:jswz} under updates to $\mD$ and $h$
becomes the task of maintaining the inverse of some matrix $\mM$.
Changing some $k$ entries in $\mD$ or $h$ corresponds to changing some $k$ entries in $\mM$.
Querying $\tilde{O}(\sqrt{n})$ entries of \eqref{eq:ipm:jswz}
corresponds to querying $\tilde{O}(\sqrt{n})$ entries of $\mM^{-1}$.
Thus we can solve this task via \Cref{lem:elementInverse}
and further the complexities stated in \Cref{lem:elementInverse} 
match exactly the bounds of the data structure constructed in \cite[Section C-E]{JiangSWZ20}\footnote{%
Specifically \textsc{Update} matches \cite[Lemma E.3]{JiangSWZ20}, 
\textsc{ResetMiddleLayer} matches \cite[Lemma E.18]{JiangSWZ20}, 
and \textsc{ResetTopLayer} matches \cite[Lemma E.12]{JiangSWZ20},
see \cite[Table 11]{JiangSWZ20}.
We would like to point out that here we do not reproduce the data structure of \cite[Appendix I]{JiangSWZ20},
which would require further modification to \Cref{lem:elementInverse}.
},
so we are able to reproduce their data structure.

The amortized analysis of Jiang et al.~results in a $\tilde{O}(n^{1+5/9})$ time bound per iteration (assuming $\omega=2$).
While we do not reproduce their amortized complexity analysis here 
(we only reproduced their data structure),
we want to give an intuitive explanation of why $\tilde{O}(n^{1+5/9})$ amortized time per iteration is achieved,
i.e.~why an $\tilde{O}(n^{2+1/18})$-time LP solver is obtained instead of an $\tilde{O}(n^{2+1/6})$-time one.

If one ignores the impact of the randomization, 
then changing $2^j$ entries happens at most once every $\tilde{\Omega}(2^{j/2})$ iterations,
similar to \Cref{task:lp}.\footnote{Because of the random matrix $\mR$ this deterministic bound does not hold. In \cite{JiangSWZ20}, Jiang et al.~instead bound the expected time complexity.}
So by using the same argument as in \eqref{eq:amortized}
we can bound the amortized cost of \textsc{Update}, \textsc{ResetMiddleLayer} and \textsc{ResetTopLayer}
by $\tilde{O}(n^{\omega(a,b,0)} + n^{\omega(a,b,b)-b/2})$,
$\tilde{O}(n^{\omega(1,a,b)-b/2} + n^{\omega(1,a,a)-a/2})$
and $\tilde{O}(n^{\omega(1,1,a)-a/2} + n^{\omega-1/2})$ respectively,
if we always perform \textsc{ResetMiddleLayer} whenever more than $n^a$ entries changed 
since the last call to \textsc{ResetMiddleLayer},
and we perform \textsc{ResetTopLayer} whenever more than $n^b$ entries changed 
since the last call to \textsc{ResetTopLayer}.

The cost of obtaining the $\tilde{O}(\sqrt{n})$ entries of the value of the formula
via \textsc{Query} is $\tilde{O}(n^{\omega(a,b,0.5)})$,
so the amortized cost per iteration is
\[
\tilde{O}(
n^{\omega(a,b,0.5)}
+ n^{\omega(1,a,b)-b/2}
+ n^{\omega(1,1,a)-a/2}
+ n^{\omega-1/2}
).
\]
If we assume $\omega=2$ and $a,b\ge 0.5$, then this becomes
\[
\tilde{O}(
n^{a+b}
+ n^{1+a-b/2}
+ n^{2-a/2}
+ n^{1.5}
)
\]
which for $a = 8/9, b = 6/9$
is $\tilde{O}(n^{1+5/9})$ amortized time,
or $\tilde{O}(n^{2+1/18})$ total time over $\tilde{O}(\sqrt{n})$ iterations.

\subsection{Online Linear System}
\label{sec:onlinematrix}

In the online linear system problem one considers a linear system of the form $\mA x = b$ 
for some non-singular $n\times n$ matrix $\mA$. 
Let $\mA_{[s],[s]}$ be the top $s \times s$ block of $\mA$ 
and $b_{[s]}$ be the top $s$ entries of $b$
and assume that $\mA_{[s],[s]}$ is non-singular for all $s=1,...,n$.
Then the task is to solve $\mA_{[s],[s]} x = b_{[s]}$ for all $s$.
The difficulty of the online linear system problem 
comes from the fact that matrix $\mA$ and vector $b$ are given in an online way,
that is, the algorithm must answer the solutions for current $s$,
before row number $s+1$ of $\mA$ and entry number $s+1$ of $b$ are given.
Storjohann and Yang \cite{StorjohannY15} have given an $O(n^\omega)$ time algorithm for this problem.
Here we show that the same result can be obtained by \Cref{thm:formula}.
For that consider formula 
\begin{align}
(\mD\mB\mD+(\mI-\mD))^{-1} \mD d.\label{eq:onlineLS}
\end{align}
Here $\mD$ is an $n\times n$ diagonal matrix,
$\mB$ is an $n\times n$ matrix,
and $d$ is an $n$ dimensional vector.

Assume that during iteration number $s$ of the online linear system problem,
we have that the first $s$ diagonal entries of $\mD$ are $1$ and the rest is $0$.
Further, the first $s$ rows of $\mB$ are exactly the first $s$ rows of $\mA$,
while all other rows of $\mB$ are $0$.
Additionally, the first $s$ entries of $d$ are the first $s$ entries of $b$,
while the remaining entries of $s$ are $0$.
Then we have
\[
(\mD\mB\mD+(\mI-\mD))^{-1} \mD d
=
\mB_{[s],[s]}^{-1} d_{[s]}
=
\mA_{[s],[s]}^{-1} b_{[s]}.
\]
Thus by maintaining \eqref{eq:onlineLS} under row updates to add the new rows/entries to $\mB$, $\mD$, and $d$ in each iteration,
we are able to solve the online linear system problem.
Note that in each iteration, we must perform $O(1)$ row updates 
to insert the extra $1$ in $\mD$, insert the new row into $\mB$,
and insert the new entry to $d$.
Further, we require only one column of the result, as the result is just a vector.
At last, observe that the position of these updates is a fixed sequence: in iteration number $s$
we only touch the $s$-th diagonal entry of $\mD$, the $s$-th row of $\mB$ and the $s$-th entry of $d$.
Thus we can use the last data structure in \Cref{thm:formula} with $O(n^{\omega-1})$ update time,
for a total time of $O(n^\omega)$, reproducing the result of Storjohann and Yang \cite{StorjohannY15}.
We further want to note that the algorithm of Storjohann and Yang 
requires $O(n^{\omega-1})$ \emph{amortized} time per iteration,
while the update time of our algorithm is \emph{worst-case} 
after an initial $O(n^\omega)$ initialization time before reading any input.\footnote{%
It is possible to remove the $O(n^\omega)$ initialization time for this application.
The preprocessing of the dynamic matrix inverse data structure 
we use here consists only of computing the inverse of the given input matrix. 
For this application, the inverse can be constructed in $O(n)$ time 
because we have $\mD = \zeromat^{(n,n)}$, $\mB = \zeromat^{(n,n)}$, $d=\zeromat^{(n,1)}$.}

\subsection{QR Decomposition}
\label{sec:qr}

The $QR$ decomposition of a $d\times n$ matrix $\mV$ with $d \le n$
consists of an orthogonal matrix $\mQ$ and an upper triangular matrix $\mR$
with $\mQ\mR=\mV$.
The most simple algorithm for obtaining the $QR$ decomposition
is the Gram-Schmidt-Orthonormalization 
(as the columns of $\mQ$ are the orthonormalzation of the columns of $\mV$).
Let $v_1,...,v_n$ be input vectors (i.e.~columns of $\mV$).
Then the algorithm repeatedly computes $v'_i = v_i - \mP_{i-1}$,
where $\mP_{i-1} v_i$ is the orthogonal projection of $v_i$ onto $span\{v_1,...,v_{i-1}\}$,
and then the vector $v'_i$ is normalized.
The vector $v'_i$ is usually computed via
\begin{align}
v'_i = v_i - \mP_{i-1}v_i = v_i - \sum_{j=1}^{i-1} v'_j v'^\top_j v_i \label{eq:gs}
\end{align}
which takes $O(n^2)$ time
and thus the entire algorithm runs in $O(n^3)$ total time.
We can speed this up to $O(n^\omega)$ time by writing $\mP_{i-1} v_i$ as a formula 
whose input matrices change in a sparse way
and then applying a data structure from \Cref{thm:formula} to maintain this formula.
The formula we use to maintain $\mP_{i-1} v_i$ is given by the following lemma.

\begin{lemma}\label{lem:gs_formula}
Let $B \subset [i-1]$ be a set of indices of a maximal independent set of the first $i-1$ columns of $\mV$,
i.e.~$span\{v_j \mid j\in B\} = span\{v_1,...,v_{i-1}\}$ and $rank(v_1,...,v_{i-1}) = |B|$.
Define $\mD_B$ to be the $n\times n$ diagonal matrix 
where the $j$-th diagonal entry is $1$ if $j\in B$ and $0$ otherwise. 
Then we have
$$
\mP_{i-1} = \mV \mD_B (\mD_B \mV^\top \mV \mD_B + (\mI - \mD_B))^{-1} \mD_B \mV^\top .
$$
\end{lemma}
\Cref{lem:gs_formula} together with \Cref{thm:main} implies an $O(n^\omega)$ time algorithm for computing the orthonormalization of $v_1,...,v_n$.
Let $B$ be initially an empty set and maintain
\begin{align}
\left(\mI - \mV \mD_B (\mD_B \mV^\top \mV \mD_B + (\mI - \mD_B))^{-1} \mD_B \mV^\top \right) \mV. \label{eq:qr}
\end{align}
During iteration $i$ we query the $i$-th column of \eqref{eq:qr} and save the result as $v'_i$.
If $v'_i$ is a non-zero vector, then we add $i$ to set $B$,
otherwise we do not add it to $B$.
If $B$ was a maximal independent set of the first $i-1$ columns of $\mV$,
then after this iteration it is also a maximal independent set of the first $i$ columns.
To see this note that $v'_i$ is non-zero if and only if $v'_i$ is linearly independent of $\{v_j \mid j \in B\}$.

Next, we increase $i$ and repeat until we have $v'_1,...,v'_n$.
By \Cref{lem:gs_formula} we computed $v'_i = v_i - \mP_{i-1} v_i$ during iteration number $i$.
Further note that the position of the updates and queries is fixed, i.e.~during iteration $i$ we update the $i$-th diagonal entry of $\mD_B$ (we either set it to $1$ or $0$) and we query the $i$-th column of \eqref{eq:qr}.
Thus we can use the data structure variant of \Cref{thm:formula}
that has $O(n^{\omega-1})$ time per update and query for a total time of $O(n^\omega)$.

This example shows how one can speed up simple iterative algorithms 
by constructing a formula (e.g.~\eqref{eq:qr}) 
that represents the required computation in each iteration of the iterative algorithm
(e.g.~\eqref{eq:gs}).

\begin{proof}[Proof of \Cref{lem:gs_formula}]
Note that for any $n\times d$ matrix $\mA$ of rank $d$,
the matrix
$\mA (\mA^\top \mA)^{-1} \mA^\top$
is the orthogonal projection onto the image of $\mA$.
Thus for set $B \subset [i-1]$ as in \Cref{lem:gs_formula} we have that
$
\mV \mI_{[n],B} (\mI_{[n],B} \mV^\top \mV \mI_{[n],B})^{-1} \mI_{B,[n]} \mV^\top
$
is the orthogonal projection onto the image of $\mV \mI_{[n],B}=\mV_{[n],B}$, i.e.~$span\{v_1,...,v_{i-1}\}$.

For $\mM := \mI_{B} \mV^\top \mV \mI_{B} + (\mI - \mI_{B})$ and $N := [n] \setminus B$ we have 
\begin{align}
(\mM^{-1})_{B,B} = (\mI_{B,[n]} \mV^\top \mV \mI_{[n],B})^{-1},~~
(\mM^{-1})_{N,N} = \mI. \label{eq:blowup}
\end{align}
To see this assume without loss of generality that $B = [k]$ for $k = |B|$ 
(otherwise we could reorder the columns of $\mV$),
in which case we have
\begin{align*}
\mI_{B} \mV^\top \mV \mI_{B} + (\mI - \mI_{B})
=
\left[\begin{array}{cc}
\mI_{B,[n]} \mV^\top \mV \mI_{[n],B} & \\
& \mI_{[n]\setminus B,[n]\setminus B}
\end{array}\right].
\end{align*}
So to invert the matrix we just have to invert the block $\mM_{B,B}$ as predicted by \eqref{eq:blowup}.
In summary we show
\[
\mV \mD_B (\mD_B \mV^\top \mV \mD_B + (\mI - \mD_B))^{-1} \mD_B \mV^\top
=
\mV \mI_{[n],B} (\mI_{[i],B} \mV^\top \mV \mI_{[n],B})^{-1} \mI_{B,[n]} \mV^\top
\]
is the orthogonal projection onto $span\{v_1,...,v_{i-1}\}$.
\end{proof}

\subsection{Insights regarding Fine Grained Complexity}
\label{sec:finegrained}

At last, we want to discuss implications of \Cref{thm:main} 
from a fine grained complexity perspective.
In the static setting (i.e.~when the input matrices do not change),
computing product, inverse, or determinant of a matrix all require $O(n^\omega)$ time \cite{BurgisserCS97}.
In fact, these problems are equivalent and can be reduced to one-another.
This also means that computing some formula $f$ consisting of $O(1)$ matrix products, sums, differences, and inversions
can be computed in $O(n^\omega)$ time.
However, in the dynamic setting (when input matrices are allowed to change),
the time required for maintaining $f$ can vary a lot.
Consider for example the task of maintaining the matrix product $\mM:=\mA\mB$ of two $n\times n$ matrices,
while entries of $\mA$ or $\mB$ change.
Every time an entry of $\mA$ (or $\mB$) changes,
we only have to add a row of $\mB$ (or column of $\mA$) to the result $\mM$,
so the problem can be solved in $O(n)$ time.
For the product $\mM := \mA\mB\mC$,
the problem suddenly becomes a lot harder,
because now changing an entry in $\mB$
can change all $n^2$ entries of $\mM$.
This shows that maintaining some formulas $f$ is easier than others
which raises the question: \\
\emph{What are the hardest formulas? Is there a limit for how hard maintaining a formula can be?}

\Cref{thm:main} shows that just maintaining the inverse of some matrix is already the hardest case,
and all other formulas can be reduced to it.

%% file: appendix.tex
\section{Appendix}
\label{sec:appendix}

\exactComplexity*

\begin{proof}
We compute $\mM^{-1}$ during the preprocessing.
The Sherman-Morrison-Woodbury identity
\cite{ShermanM50,Woodbury50}
predicts for invertible $\mM$ and $\mM+\mU\mV^\top$ that
\ifdefined\SODAversion
$$
(\mM+\mU\mV^\top\hspace{-1pt})^{-1}\hspace{-2pt} = \mM^{-1}\hspace{-1pt} - \mM^{-1}\mU (\mI + \mV^\top\hspace{-1pt} \mM^{-1} \mU)^{-1} \mV^\top \hspace{-1pt}\mM^{-1}
$$
\else
\[
(\mM+\mU\mV^\top)^{-1} = \mM^{-1} - \mM^{-1}\mU (\mI + \mV^\top \mM^{-1} \mU)^{-1} \mV^\top \mM^{-1}
\]
\fi
When changing $n^a$ entries in $n^b$ columns, 
then $\mU,\mV$ can be written as $n \times n^b$ matrices.
Specifically when adding some matrix $\Delta$ to $\mM$,
then we can write $\Delta = \mU\mV^\top$ 
where $\mU$ consists of the nonzero columns of $\Delta$
and $\mV$ consists of the column positions, 
i.e.~if the $k$-th non-zero column of $\Delta$ is the $j$-th column,
then $\mV_{[n],k} = \unitvec_j$ (the $j$-th standard unit vector).

During \textsc{Update}, we only compute the inverse $(\mI + \mV^\top \mM^{-1} \mU)^{-1}$.
Here $\mV^\top \mM^{-1} \mU$ are $n^b$ columns of $\mM^{-1} \mU$
because of the sparsity of $\mV$,
so we can compute the matrix in $O(n^{a+b})$ operations.
Inverting the resulting matrix requires $O(n^{b\cdot\omega})$ operations.

During \textsc{Query}, we want to compute 
$$(\mM^{-1})_{I,j} - (\mM^{-1} \mU (\mI + \mV^\top \mM^{-1} \mU)^{-1} \mV^\top \mM^{-1} \unitvec_j)_{I}.$$
We start by computing the vector
$v := (\mI + \mV^\top \mM^{-1} \mU)^{-1} \mV^\top \mM^{-1} \unitvec_j$
by multiplying from right to left: 
$\mV^\top \mM^{-1} \unitvec_j$ are just $n^b$ entries of a column of $\mM^{-1}$
so computing this product can be done in $O(n^b)$ time.
Multiplying the vector with the inverse 
$(\mI + \mV^\top \mM^{-1} \mU)^{-1}$
requires $O(n^{2b})$ time.
Note that the inverse was already computed during \textsc{Update}.

Next, we compute the rows given by $I$ of $\mM^{-1} \mU$
in $O(|I|\cdot n^a)$ operations,
and at last we multiply these rows with the vector
$v$ to obtain $(\mM^{-1} \mU v)_I$.
This step takes $O(|I|n^{b})$ time which is subsumed by $O(|I|n^a)$.
Subtracting the result from $(\mM^{-1})_{I,j}$ is the desired answer for the call to \textsc{Query}.

For \textsc{Reset} we compute $(\mM+\mU\mV)^{-1}$ explicitly via the Sherman-Morrison-Woodbury identity in $O(n^{\omega(1,1,b})$ operations and then set $\mM \leftarrow \mM + \mU\mV^{-1}$.
\end{proof}